\documentclass[a4paper,UKenglish,cleveref, hyperref,autoref, thm-restate,lstlisting]{lipics-v2019}

\usepackage{pgfplots}

\title{Concentration Bounds for the Collision Estimator} 

\titlerunning{Collision Estimator} 

\author{Maciej Skorski}{University of Luxembourg}{maciej.skorski@gmail.com}{}{}

\authorrunning{M. Skorski} 

\Copyright{M., Skorski} 

\ccsdesc[500]{Theory of computation~Approximation algorithms analysis}

\keywords{Entropy Estimation, Collision Estimation, Birthday Paradox} 

\category{} 

\relatedversion{} 

\supplement{}

\funding{}

\nolinenumbers 

\EventAcronym{Submitted}
\EventYear{2020}
\ArticleNo{\ }

\usepackage{tikz}
\usepackage{pgfplots}
\pgfplotsset{compat=1.13}
 
\crefname{lstlisting}{listing}{listings}
\Crefname{lstlisting}{Listing}{Listings}

\hypersetup{
    colorlinks = true
}

\begin{document}

\maketitle

\begin{abstract}
We prove a strong concentration result about the natural collision estimator, which counts the number of collisions that occur within an iid sample.
This estimator is at the heart of algorithms used for uniformity testing and entropy assessment.

While the prior works were limited to only variance, we use elegant techniques of independent interest to bounds higher moments and conclude concentration properties.
As an immediate corollary we show that the estimator achieves high-probability guarantee on its own and there is no need for boosting (aka median/majority trick).
\end{abstract}

\section{Introduction}

\subsection{Collision Estimation}

For many applications, such as key derivation in cryptography~\cite{dodis2013overcoming}, property testing~\cite{goldreich2017introduction} and general algorithms~\cite{acharya2014complexity}
 it is of interest to estimate the \emph{collision probability} of a distribution $X$ 
\begin{align}\label{eq:cp}
Q \triangleq \sum_x \Pr[X=x]^2
\end{align}
Given a sample $X_1,\ldots,X_n\sim^{iid} X$ one defines the "natural" collision estimator as
\begin{align}\label{eq:cp_sample_estimator}
\tilde{Q} \triangleq \frac{1}{n(n-1)}\sum_{i\not=j}\mathbb{I}(X_i=X_j).
\end{align}
which resembles the birthday paradox. In this work we obtain a strong result about its concentration properties, which can be formally stated as follows.

\subsection{Our Contribution}

\begin{theorem}[Tails of Collision Estimator]\label{thm:main}
The estimator \eqref{eq:cp_sample_estimator}, after centering, has tails
\begin{align}
\Pr[|\tilde{Q}-Q|>\epsilon] \leqslant O(1)\exp( -\Omega( \min( \epsilon^2 /  v^2, \epsilon/b , n\sqrt{\epsilon }) ) ).
\end{align}
where we define
\begin{align}
v^2 & \triangleq \sum_x\Pr[X=x]^2/n^2 + \sum_x\Pr[X=x]^3/n  \\
b & \triangleq \max_x \Pr[X=x] / n
\end{align}
\end{theorem}

\begin{remark}[Intuition: variance and scale]
Best way to understand the concentration bounds in \Cref{thm:main} is to think of $v^2$ as a variance proxy (in fact we have $\mathbf{Var}[\tilde{Q}] = O(v^2)$)
 and of $b$ as as scale parameter. Then the tail of $\exp( -\Omega(\min(\epsilon^2/v^2,\epsilon/b ))$ is typical for so-called sub-gamma distributions~\cite{boucheron2013concentration}. The term with $\epsilon^{1/2}$ appears due to a possible \emph{heavy tail behavior}: when the moments grow like $d^{2d}$ we get the tail of $\mathrm{e}^{-\Omega(\epsilon^{1/2})}$.
\end{remark}

\subsubsection{Related Work}
To the best knowledge of the author, there are no prior works on \emph{exact concentration of the collision estimator}.
The variance of collision estimator itself has been studied extensively in the context of \emph{uniformity testing}~\cite{batu2001testing,goldreich2011testing,paninski2008coincidence,diakonikolas2016collision,goldreich2017introduction},
 and Renyi entropy estimation~\cite{acharya2014complexity,acharya2016estimating,obremski_et_al:LIPIcs:2017:7569}, but we lack of understanding of higher moments
and concentration properties. The techniques used to handle the variance were merely manipulation of algebraic expressions with some combinatorics to carry out term cancellations, which is hard to scale to higher moments. 
It is also not possible to derive a concentration result by a black-box application of known concentration inequalities: the main problem is that 
the estimator $\tilde{Q}$ is a quadratic form of correlated inputs, where the inputs are possibly very rare events. Leaving aside the problem of correlation,
the right tool to attack the quadratic form would be the Hanson-Wright inequality; however examining the state-of-art variants~\cite{rudelson2013hanson,bellec2019concentration} we find
them insufficient in our context (for example, we get very weak scale term $b$). For these reason we resort to direct moment estimates; again there is no directly
applicable formula, but at the core of our proof is the sharp moment characterization due to Latala~\cite{latala1997estimation}.



\subsubsection{Outline of Techniques}

Our result in \Cref{thm:main} is of interest not only because of its strong quantitative guarantees (which we will see later when discussing applications) but also because of elegant techniques of independent interest, that are used in the proof. We elaborate on that below.

\paragraph{Negative Dependence}

Negatively dependency of random variables, roughly speaking, captures the property that one of them increases others are more likely to decrease. 
This property is a very strong form of negative correlation known to imply concentration bounds comparable to those of independent random variables; essentially (in the context of concentration bounds) one works with negatively dependent random variables as if they were independent, which simplifies an analysis to a great extent. Best-known from applications to balls and bin problems, the theory of negative dependence has been summarized in~\cite{joag1983negative} and~\cite{dubhashi1996balls}.

In this work we leverage the negative dependence by proving that this property holds for estimator's contributions (thought as loads) from possible outcomes of the distribution  (thought as bins);
more precisely negatively dependent are $\tilde{Q}_x\triangleq \sum_{i\not=j}\mathbb{I}(X_i = X_j = x)$ indexed by $x$.
This reduces the problem to studying sums of independent random variables distributed as $\tilde{Q}_x$.
As a remark we note that this trick can be also used to simplify a bulk of computations in the case of higher-order collisions, studied in higher-order Renyi entropy estimators~\cite{acharya2016estimating,acharya2014complexity,obremski_et_al:LIPIcs:2017:7569}.

\paragraph{Subtle Moment Methods}

Most of concentration results in TCS papers are obtained by a black-box application of Chernoff-like bounds, and it is not so common to face up a case where these inequalities fail to produce good results. As we point out in this work, collision estimation seems to be such a use case. The problem is that observing every fixed element $x$ in a sample is a \emph{rare event} with extremely small probability (for example, in the birthday paradox setup we have $\Pr[X=x] = 1/m$ whereas $n = O(\sqrt{m})$). 
Since we have $\tilde{Q}_x=\sum_{i\not=j}\mathbb{I}(X_i =x)\mathbb{I}( X_j = x) $ which are \emph{quadratic forms} of these rare events, we may want to apply
a variant of Hanson-Wright's Inequality~\cite{rudelson2013hanson,bellec2019concentration} and then to assemble obtained concentrations of $\tilde{Q}_x$ into a concentration result for $
\sum_{x}\tilde{Q}_x$ (e.g. by Cramer-Chernoff); unfortunately best known bounds for quadratic forms do not behave well if inputs are very small
(they contain distribution-free terms in exponent; in our case we would end up with much weaker $b = O(1/n)$). 


This motivates the broader question on what to do when exponential inequalities fail? Our solution is to resort to subtle moment methods that have been studied particularly by Latala~\cite{latala1997estimation}.
One important contributions of this paper is that we simplify one of his results, showing how to estimate moments (and hence concentration properties) of sums $\sum_x Z_x$  by \emph{controlling sum of moments}, e.g. by bounding $\sum_x \mathbf{E} |Z_x|^d$. This bound is very convenient as individual moments are much easier to compute; we call such conditions~\emph{Rosenthal-type} due to the celebrated result of Rosenthal~\cite{rosenthal1970subspaces} in same spirit.
We note that our technique can be used in other problems where applications of exponential concentration bounds are problematic, for example applied to the problem of missing mass~\cite{ortiz2003concentration}.

\paragraph{Reduction to Binomial Moments}

Armed with the Rosenthal-type concentration bounds we are left with estimating the moments of estimator contributions $\tilde{Q}_x$. 
Here we apply the tricks that have been proven useful when dealing with quadratic forms:
 centering, symmetrization and decoupling. Eventually we are able to link moments of $\tilde{Q}_x$ with those of (symmetrized) binomial distributions. More precisely 
for $p=\Pr[X=x], S,S'\sim^{iid}\mathsf{Binom}(n,p)$
we obtain the following bound
$$
\mathbf{E} |\tilde{Q}_x-\mathbf{E}\tilde{Q}_x|^d \leqslant O( \mathbf{E} |S-S'|^d )^2 + O(d n^2p^2)^{d/2} \mathbf{E} |S-S'|^d 
$$
This explains the specific form of \Cref{thm:main}: consider for simplicity $d=2$ then the first and second term on the right-hand side contribute respectively $p=\Pr[X=x]^2$ and $p=\Pr[X=x]^3$;
with higher $d$ they contribute respectively $p^{d/2}$ and $p^{3d/2}$.

The crucial step here is to use asymptotically sharp bound on binomial moments (so that we have optimal dependency on $d$). Since they are hard to find in the literature, we prove such bounds using an elementary combinatorial approach along with symmetrization.


\subsection{Applications}

\subsubsection{Application to Uniformity Testing}

In uniformity testing one wants to know how close is some unknown distribution to the uniform one, based on a random sample.
An appealing idea is to relate the closeness to the collision probability $Q$: for a distribution over $m$ elements the smallest value of $Q$ is $1/m$ which is realized by the uniform distribution $U_m$.
Then the closer is $Q$ to $1/m$, the closer is $X$ to $U_m$. Such a test was studied by a number of authors
\cite{batu2001testing,goldreich2011testing,paninski2008coincidence,diakonikolas2016collision,goldreich2017introduction} with optimal bounds found in~\cite{diakonikolas2016collision}. Remarkably, our concentration bounds imply that such a test achieves high-probability guarantees on its own, as stated in \Cref{thm:tester}.
Prior to our work the test guarantees were quite weak so it was necessary to run multiple tests in parallel.
\lstset{
language=Python,
basicstyle=\small,
stringstyle = \ttfamily,
keywordstyle=\color{black}\bfseries,
mathescape=true,
morecomment=[s]{/*}{*/}
}
\begin{lstlisting}[label={lst:tester},caption={Uniformity Tester for Discrete Distributions.},captionpos=[b]]
def l2closeness_to_uniform($X$,$n$,$\epsilon$):
  $m\gets |\mathrm{dom}(X)|$ /* domain size */
  $x[1]\ldots x[n]  \gets^{IID} X$ /* get iid samples */
  $Q \gets \#\{ (i,j) : x[i]=x[j], i\not=j\}  $ /* count collisions */
  $Q \gets {Q}/{n(n-1)}$ /* normalize */
  if $Q>(1+\epsilon)/m$:
    return False
  elif $Q < (1+\epsilon)/m$:
    return True
\end{lstlisting}

\begin{corollary}[Optimal Sublinear Collision Tester]\label{thm:tester}
If $X$ is distributed over $m$ elements, with 
$$
n = O(\log(1/\delta) m^{1/2} / \epsilon)
$$
samples the algorithm in \Cref{lst:tester} distinguishes with probability $\delta$ between a)
$\|\mathbf{P}_X-U_m\|_2^2\leqslant \frac{\epsilon}{2m} $ and b) $\|\mathbf{P}_X-U_m\|_2^2\geqslant \frac{2\epsilon}{m} $, when $1/\sqrt{m}\leqslant \epsilon \leqslant 1$.
\end{corollary}

\begin{remark}[Comparison with ~\cite{diakonikolas2016collision}]
The novelty is that we do not require parallel runs to get small error probability $\delta$. The sample size $n$ matches 
the best bound due to ~\cite{diakonikolas2016collision} under the mild restriction that $n=O(m)$ (which implies $\epsilon = \Omega(m^{-1/2})$ in \Cref{thm:tester})
that is the number of samples is at most linear in the alphabet size.
Such sublinear algorithms are of practical interest when the alphabet is huge; the restriction $\epsilon = \Omega(m^{-1/2})$ is also sufficient for virtually all cryptographic applications because
when $X\in \{0,1\}^d$ we have that $m^{-1/2} = 2^{-d/2}$ corresponds to exponential security guarantees).

\end{remark}


\subsubsection{Application to R\'{e}nyi Entropy Estimation}
Consider the problem of \emph{relative} estimation, where $\epsilon:= \epsilon Q$. This can be seen as estimation of collision entropy $\mathbf{H}_2(X)\triangleq -\log Q$
within an additive error of $\epsilon$~\cite{acharya2016estimating}. Our result implies again that the estimator achieves high-probability guarantee on its own, without
parallel runs. 
\begin{corollary}[Collision Estimation]\label{cor:collision_estimation}
We have 
\begin{align*}
\Pr[|\tilde{Q}-Q| > \epsilon Q] \leqslant O(1)\exp(-\Omega( n \epsilon^2 Q^{1/2})),\quad 0<\epsilon < 1.
\end{align*}
\end{corollary}
Which shows that the estimator requires 
\begin{align*}
n =O( \log(1/\delta)Q^{-1/2} /\epsilon^2) 
\end{align*}
samples to achieve relative error of $\epsilon$ and probability guarantee of $1-\delta$.
\begin{remark}[Optimality]
In the worst case we have $n = O( m^{1/2}/\epsilon^2)$ when the domain of $X$ has $m$ elements, which matches the lower bounds~\cite{acharya2016estimating}.
\end{remark}

\begin{remark}[Difference from Uniformity Testing: Phase Transition]
The bounds for collision estimation are sharp when $Q = \Omega(1/m)$, however for uniformity testing one considers the different regime
of $Q = 1/m\cdot (1+o(1))$, so the lower bounds~\cite{acharya2016estimating} no longer apply.
Indeed, uniformity testing allows a better dependency on $\epsilon$ that suggested by the general collision estimation. This is an interesting phenomena that could be seen as a "phase transition".
\end{remark}

\subsection{Organization}

In \Cref{sec:apprs} we show in detail how to derive results on applications claimed above.
The proof of the main result is given in \Cref{sec:proof} and follows the presented outline.
In \Cref{sec:conclude} we conclude the work.







\section{Preliminaries}

\subsection{Sub-Gamma Distributions}
A random variable $Z$ is \emph{sub-gamma} with variance factor $v^2$ and scale $b$ when~\cite{boucheron2013concentration}
$$\mathbf{E}\exp(tZ) \leqslant \exp\left(\frac{v^2t^2}{2(1-bt)}\right),\text{ when } |t|<1/b$$
Such a distribution has gamma-like tails (by the Cramer-Chernoff method~\cite{cramer1938nouveau,chernoff1952measure}, see~\cite{boucheron2013concentration})
\begin{proposition}[Sub-Gamma Tails]
If $Z$ is sub-gamma with variance factor $v^2$ and scale $b$ 
$$
\Pr||Z|>t]\leqslant 2\exp\left(-\frac{t^2}{2(v^2+b\cdot t)}\right),\quad t>0.
$$
\end{proposition}
Sub-gamma property aggregates when taking sums of independent variables~\cite{subg_lecture_notes}
\begin{proposition}[Sub-Gamma Aggregation]
let $Z_i$ be sub-gamma with variance factor $v_i^2$ and scale $b_i$, then $\sum_i Z_i$
is sub-gamma with parameters $v^2 = \sum_i v_i^2$ and $b= \max_i b_i$.
\end{proposition}
The sub-gamma property can be verified by the moments (see~\cite{boucheron2013concentration}, Theorem 2.3)
\begin{proposition}[Sub-Gamma Property via Moments]\label{prop:subg_moments}
Let $Z$ be centered. If $Z$ is sub-gamma with variance factor $v^2$ and scale $b$ then $(\mathbf{E}|Z|^d)^{1/d} = O(d^{1/2}v + d b  )$ for every even $d\geqslant 2$. Conversely,
when $(\mathbf{E}|Z|^d)^{1/d} = O(d^{1/2}v + d b  )$ for every even $d\geqslant 2$ then $Z$ is sub-gamma with variance factor $O(v^2+b^2)$ and scale $O(b)$.
\end{proposition}

\subsection{Rosenthal-type Moment Bounds}

The following result bounds the moments of a sum of random variables by controlling moments of individual components.
\begin{lemma}[Sharp Bounds for Moments of Independent Sums~\cite{latala1997estimation}]
For $W_x$ independent
$$
(\mathbf{E} | \sum_{x}W_x |^d)^{1/d}  = \Theta(1)\cdot \inf\{  T: \sum_x \log \mathbf{E}|1+W_x/T|^d \leqslant d  \}
$$
holds for any real $d\geqslant 1$.
\end{lemma}

\begin{lemma}[Simplified Latala's Bound]\label{prop:latala_simple}
Let $W_x$ be independent and centered, let $d$ be even and define the function $\phi(u)\triangleq \frac{(1+u)^d+(1-u)^{d}}{2}-1$. Then
$$
(\mathbf{E} | \sum_{x}W_x |^d)^{1/d}  \leqslant O(1)\cdot \inf\{  T: \sum_x \mathbf{E}\phi(W_x/T) \leqslant d  \}
$$
\end{lemma}
\begin{proof}[Proof of \Cref{prop:latala_simple}]
By Jensen's inequality
$$
\sum_x \log \mathbf{E}|1+W_x/T|^d \leqslant  m\log( m^{-1}\sum_x \mathbf{E}|1+W_x/T|^d)
$$
By the symmetrization trick we can assume that $W_{x}$ are symmetric, loosing a factor $O(1)$ in the upper bound. Expanding the $d$-th power and computing moments we obtain
$$
\frac{1}{m}\sum_x \mathbf{E}|1+W_x/T|^d  = 1+\frac{1}{m}\sum_{x}\mathbf{E}\phi(W_x/T)
$$
and since $1+\frac{1}{m}\sum_{x}\mathbf{E}\phi(W_x/T) \leqslant \exp(\frac{1}{m}\sum_{x}\mathbf{E}\phi(W_x/T) )$ we finally obtain
$$
\sum_x \log \mathbf{E}|1+W_x/T|^d  \leqslant \sum_{x}\mathbf{E}\phi(W_x/T)
$$
which finishes the proof.
\end{proof}

\subsection{Growth of Binomial Moments}

\begin{lemma}[Symmetrized Binomial Moments]\label{lemma:symm_binom_moment}
Let $S\sim \mathsf{Binom}(n,p)$ and $S'$ be an independent copy of $S$. Then letting $\sigma^2=2p(1-p)$ we have for any even positive $d$
$$
\mathbf{E}(S-S')^d \leqslant O(d)^{d/2}\sum_{\ell=1}^{d/2}\binom{n}{\ell} \ell^{d/2} \sigma^{2\ell}.
$$
\end{lemma}
\begin{note}
When $n\sigma^2= \Theta(1)$ this can grow as $O(d)^d$.
\end{note}
\begin{proof}[Proof of \Cref{lemma:symm_binom_moment}]
Let $\eta,\eta'$ be independent distributed as $\mathsf{Bern}(p)$, we can write $S-S'= \sum_{i=1}^{n}\eta_i$ where $\eta_i\sim^{iid}\eta-\eta'$. Consider now the multinomial expansion for even $d$
$$
\mathbf{E}|S-S'|^d = \mathbf{E}(\sum_{i=1}^{d}\eta_i)^d = \sum_{i_1,\ldots,i_d}\mathbf{E}[\prod_{k=1}^{d}\eta_{i_k}]
$$
Say that the tuple $(i_1,\ldots,i_d)$ has $\ell$ distinct values which appear with multiplicities $c_1,\ldots,c_k>0$, $c_1+\ldots+c_k=d$. Then $\mathbf{E}[\prod_{k=1}^{d}\eta_{i_k}] = \prod_{k=1}^{\ell} \mathbf{E}(\eta-\eta')^{c_k}$. Since $\eta-\eta'$ is symmetric we can consider only  the case where all $c_k$ are even and since $|\eta-\eta'|\leqslant 1$
we have for each $k$
$$\mathbf{E}(\eta-\eta')^{c_k} \leqslant \mathbf{E}(\eta-\eta')^2 = \mathbf{Var}[\eta-\eta'] = 2p(1-p)$$
 The number of such tuples is $\binom{n}{\ell}\binom{d}{c_1,\ldots,c_k}$, and therefore
$$
\mathbf{E}(\sum_{i=1}^{d}\eta_i)^d  \leqslant \sum_{\ell=1}^{d/2}\binom{n}{\ell}\sum_{c_1,\ldots,c_{\ell}: \text{even,positive}}\binom{d}{c_1,\ldots,c_{\ell}} \sigma^{2\ell}
$$
We are left with the combinatorial problem of determining the sum of even multinomial coefficients.
It is known that this quantity equals the moment of a rademacher sum
$$
\sum_{c_1,\ldots,c_{\ell}: \text{even,positive}}\binom{d}{c_1,\ldots,c_{\ell}}  = \mathbf{E}\left(r_1+\ldots + r_{\ell}\right)^d,\quad r_i \sim^{iid} \pm 1\ w.p.\ \frac{1}{2}.
$$
By Khintchine's Inequality this is at most $O(d\ell)^{d/2}$ and the proof is finished.
\end{proof}

\section{Applications}\label{sec:apprs}

\begin{proof}[Proof of \Cref{thm:tester}]
Since $\|\mathbf{P}_X-U_m\|^2 =\| \mathbf{P}_X\|^2 - 1/m = Q - 1/m $, 
the two cases in \Cref{thm:tester} are equivalent to a) $Q  \leqslant \frac{1+\epsilon/2}{m}$ and b) $Q \geqslant \frac{1+2\epsilon}{m}$. 
Let $\alpha$ be such that $Q = (1+\alpha)/m$, then it suffices to prove that \Cref{lst:tester} estimates $Q$ estimates
within an additive error $\max(\epsilon,\alpha)/2m$ and correctness probability $1-\delta$. 
To this end we show that the exponent in \Cref{thm:main} with $\epsilon:=\max(\epsilon,\alpha)/2m$ for $n = \Theta(n\log(1/\delta)/\epsilon$ is $\Omega(\log(1/\delta))$. In fact we show
\begin{align*}
\min\left(\frac{\max(\epsilon,\alpha)^2}{ m^2  v^2}, \frac{ \max(\epsilon,\alpha) }{ m b},  n\sqrt{ \frac{\max(\epsilon,\alpha)}{ m}}\right)= \Omega(\log(1/\delta)),\quad m^{-1/2}\leqslant \epsilon \leqslant 1.
\end{align*}
Observe that $n\sqrt{\max(\epsilon,\alpha) / m})\geqslant n\sqrt{\epsilon/m} \geqslant n m ^{-1/2}\epsilon$ because $\epsilon\leqslant 1$.
Moreover we have $n b = \Pr[X=x]\leqslant (\sum_x\Pr[X=x]^2)^{1/2} \leqslant \sqrt{(1+\alpha)/m} $,
so that $\max(\epsilon,\alpha)  /m b = \Omega(1)n m^{-1/2}  \max(\epsilon,\alpha)  / \sqrt{1+\alpha} \geqslant \Omega(1) n m^{-1/2}\epsilon$, the inequality
is immediate when $\alpha \leqslant 1$ and for $\alpha>1$ follows because $\alpha/\sqrt{1+\alpha} = \Omega(1)$ while $\epsilon\leqslant 1$. In both cases the exponent is at least 
$\Omega(\log(1/\delta))$, therefore it remains to prove this in the last case
\begin{align*}
\max(\epsilon,\alpha)^2 /m^2  v^2= \Omega(\log(1/\delta)),\quad m^{-1/2}\leqslant \epsilon \leqslant 1.
\end{align*}
We have $\alpha/m = \sum_{x}(\Pr[X=x]-1/m)^2$ by the definition of $\alpha$, thus
$\sum_x\Pr[X=x]^3 = \frac{1}{m^2} +\frac{ 3\alpha}{m^2} + \sum_x(\Pr[X=x]-1/m)^3$; we can now bound 
$\sum_x\Pr[X=x]^3\leqslant (1+3\alpha)/m^2 + (\alpha/m)^{3/2}$.
and consequently $v^2 \leqslant (1+\alpha)/m n^2  + ( (1+3\alpha)/m^2 + (\alpha/m)^{3/2})/n$. Therefore
\begin{align*}
\max(\epsilon,\alpha)^2 /m^2  v^2 \geqslant \frac{1}{3}\min\left( n^2 m^{-1}\frac{\max(\epsilon,\alpha)^2}{1+\alpha},  n\frac{\max(\epsilon,\alpha)^2}{1+3\alpha},  n m^{-1/2} \frac{\max(\epsilon,\alpha)^2}{\alpha^{3/2}} \right)
\end{align*}
Since $\frac{\max(\epsilon,\alpha)^2}{1+\alpha}=\Omega(\epsilon^2)$, $\frac{\max(\epsilon,\alpha)^2}{1+3\alpha} = \Omega(\epsilon^2)$
and $ \frac{\max(\epsilon,\alpha)^2}{\alpha^{3/2}} = \Omega(\epsilon^{1/2}) = \Omega(\epsilon)$ for $\epsilon \leqslant 1$
\begin{align*}
\max(\epsilon,\alpha)^2 /m^2  v^2 \geqslant \Omega(1)\min\left( n^2 m^{-1}\epsilon^2,  n\epsilon^2,  n m^{-1/2} \epsilon^{1/2}\right)
\end{align*}
Finally we use the assumption $\epsilon \geqslant m^{-1/2}$ which gives us
\begin{align*}
\max(\epsilon,\alpha)^2 /m^2  v^2 \geqslant \Omega(1)\min( (n m^{-1/2})^2,  n m^{-1/2}\epsilon) = \Omega(\min(\log^2(1/\delta),\log(1/\delta))) 
\end{align*}
which finishes the proof because $\delta<1$ and $\log^2(1/\delta)\geqslant \log(1/\delta)$.
\end{proof}

\begin{proof}[Proof of \Cref{cor:collision_estimation}]
Observe that we can bound $\sum_{x}\Pr[X=x]^3< (\Pr[X=x] \sum_x\Pr[X=x]^2)^{3/2}$ 
which implies $v^2 \leqslant Q/n^2 + Q^{3/2}/n$ and $b\leqslant Q^{1/2}/n$. 
Therefore the exponent in \Cref{thm:main} is at least $\Omega(1)\min(\epsilon^2 Q^2/v^2, \epsilon Q / b, n\sqrt{\epsilon Q}) = \Omega(1)\min(n^2 Q \epsilon^2,n Q^{1/2}\epsilon^2,\epsilon n Q^{1/2},n\sqrt{\epsilon Q})$, and the claim on collision estimation follows because $\epsilon \leqslant 1$.
\end{proof}

\section{Proof of Main Result}\label{sec:proof}

\subsection{Collision Estimator as Function of Histogram}

The first trick is to condition on possible values $x$ in the sample. We have
\begin{align}
    \tilde{Q} = \frac{1}{n(n-1)}\sum_x \sum_{i\not=j} \mathbb{I}(X_i=x)\mathbb{I}(X_j=x)
\end{align}
Next we decompose the estimator into the sum of contributions from different $x$
\begin{align}\label{eq:estimator_sum}
    \tilde{Q} = \frac{1}{n(n-1)}\sum_{x}\tilde{Q}_x,\quad \tilde{Q}_x =  S_x^2-S_x,\quad S_x =\sum_i \mathbb{I}(X_i=x)
\end{align}
which is the relation to the histogram of the sample $X_1,\ldots,X_n$, as  $S_x$ is the load of bin $x$.
Observe also that $\tilde{Q}_x/n(n-1)$ is, for each $x$, an unbiased estimator for $\Pr[X=x]^2$.

\subsection{Utilizing Negative Dependence}

\begin{lemma}[Contributions from Bins are Negatively Dependent]\label{claim:neg_dependent}
Random variables $\{S^2_x-S_x\}_x$, and therefore $\tilde{Q}_x$ (defined in \Cref{eq:estimator_sum}) are negatively dependent.
\end{lemma}
\begin{proof}
Observe that for any fixed $i$ the random variables $\mathbb{I}(X_i=x)$, indexed  by $x$, are negatively dependent because they are boolean and add up to one (zero-one property, see  Lemma 8 in~\cite{dubhashi1996balls}).
Since $X_i$ for different $i$ are independent, we obtain that $(\mathbb{I}(X_i=x))_{i,x}$ indexed by \emph{both} $i$ and $x$ are negatively dependent (augmentation property, see Proposition 7 part 1 in~\cite{dubhashi1996balls}).
Observe that $S_x^2-S_x = f((\mathbb{I}(X_i=x))_i)$ with $f(u_1,\ldots,u_n) = \sum_{i\not=j} u_i u_j$ increasing in each $u_i$ when $u_i\geqslant 0$. Applying increasing functions to non-overlapping subsets of negatively dependent variables produces
variables that are also negatively dependent (aggregation by monotone functions, see Proposition 7 part 2 in~\cite{dubhashi1996balls}), therefore $S^2_x-S_x$ are negatively dependent. Same holds for $\tilde{Q}_x$ which differ only by a scaling factor.
\end{proof}

\subsection{Concentration in Single Bins}

We will now study the properties of $\tilde{Q}_x$, for each fixed value of $x$.  For brevity we denote 
$$
 p = \Pr[X=x]
$$

\subsubsection{Centering}
We start by centering random variables $\mathbb{I}(X_i=x)$. Direct calculations show that
\begin{proposition}[Estimator Bin Contributions]\label{sec:bin_contrib}
The centered contribution from bin $x$ is
\begin{align}\label{eq:bin_contributions}
S^2_x - S_x - \mathbf{E}[S^2_x-S_x] = U_2+ 2(n-1)p \cdot U_1
\end{align}
where $U_1$ and $U_2$ are zero-mean given by 
\begin{align}
U_1 = \sum_i \xi_i , \quad U_2 = \sum_{i\not=j}\xi_i\xi_j,\quad \xi_i = \mathbb{I}(X_i=x)-\Pr[X=x]
\end{align}
\end{proposition}
\begin{proof}
Let $Z_ i = \mathbb{I}(X_i=x)$, then $S^2_x-S_x - \mathbf{E}[S^2_x-S_x] = \sum_{i\not=j}(Z_i Z_j-\mathbf{E}[Z_i Z_j])$ and $\mathbf{E}[Z_i Z_j] = p^2$ when $i\not=j$.
The result follows from the identity $Z_i Z_j - p^2 = (\xi_i+p)(\xi_j+p)-p^2 = \xi_i\xi_j + p(\xi_i+\xi_j)$, summed over pairs $i\not=j$.
\end{proof}

\begin{note}[Symmetric Polynomials]
Observe that $U_1$ and $U_2$ are the first and second elementary symmetric polynomials in variables $\xi_i$
\end{note}

\subsubsection{Bounding Variance}

\begin{corollary}[Total Variance of Collision Estimator]\label{cor:variance_full_analysis}
The contribution from bin $x$ satisfies
$$\mathbf{Var}[\tilde{Q}_x] = O(p^2/n^2+p^3/n^3),\quad p=\Pr[X=x]$$
and the total variance of the collision estimator is
$$
\mathbf{Var}[\tilde{Q}] \leqslant \sum_{x}\mathbf{Var}[\tilde{Q}_x] = O( \sum_{x}\Pr[X=x]^2/n^2 + \sum_{x}\Pr[X=x]^3/n^3)
$$
\end{corollary}


\begin{proof}
By inspection of \Cref{eq:bin_contributions} we see that $U_1$ and $U_2$ are uncorrelated so that
$$
\mathbf{Var}[\tilde{Q}_x] = \mathbf{Var}[U_1] + \Theta ((np)^2)\mathbf{Var}[U_2]
$$
Easy inspection shows $\mathbf{E}[U_2^2] = n(n-1)p^2 $ and $\mathbf{E}[U_1^2] = np$;
since $\mathbf{E}[U_2]=\mathbf{E}[U_1]=0$ this shows $\mathbf{Var}[\tilde{Q}_x] = \Theta( (np)^2 + (np)^3)$.
The total variance bound follows because by negative dependence 
$
\mathbf{Var}[\sum_x  \tilde{Q}_x] \leqslant \sum_x \mathbf{Var}[\tilde{Q}_x]$.
\end{proof}

\subsubsection{Bounding Moments by Decoupling and Symmetrization}

We will bound higher moments of $U_2$ and $U_1$ in \Cref{eq:bin_contributions} in terms of binomial moments.

The following is a well-known decoupling inequality (cf. Theorem 6.1.1 in~\cite{vershynin2018high})
\begin{proposition}[Decoupling for Quadratic Forms]\label{prop:decouple}
Let $\xi = (\xi_1,\ldots,\xi_n)$  be a random vector with centered independenet components,
let $A = a_{i,j}$ be a \emph{diagonal-free} matrix of shape $n\times n$. Then for any
convex function $f$
$$
\mathbf{E} f(\xi^TA\xi) \leqslant 4\mathbf{E}f(\xi^T A \xi)
$$
where $\xi'$ is independent and identically distributed  as $\xi$.
\end{proposition}
We also need the following standard fact on symmetrization (cf. Lemma 6.1.2 in~\cite{vershynin2018high})
\begin{proposition}[Symmetrization Trick]\label{prop:symmetrize}
Let $Y,Z$ be independent and $\mathbf{E}Z=0$, then $\mathbf{E}f(Y)\leqslant \mathbf{E}f(Y+Z)$ for any convex $f$.
\end{proposition}
\begin{note}
These are crutial for proving Hanson-Wright's Lemma.
\end{note}
By combining \Cref{prop:decouple} and \Cref{prop:symmetrize} we obtain
that when calculating the moments of $U_2$ we can assume (loosing a constant factor) that
$\xi_i$ are symmetric and decoupled. 
\begin{lemma}[Bounding Quadratic Contributions]\label{lemma:quadratic_contrib}
For $\xi_i\sim^{iid} \mathsf{Bern}(p)-p$ and even $d$
$$
\mathbf{E}|\sum_{i\not=j}\xi_i\xi_j|^d \leqslant 4\cdot \mathbf{E}|\sum_{i\not=j} \eta_i \eta'_j|^d,\quad \eta_1,\eta_2\ldots,\eta'_1,\eta'_2\ldots\sim^{iid} \eta-\eta',\quad \eta,\eta'\sim \mathsf{Bern}(p)
$$
\end{lemma}
\begin{note}
The off-diagonal assumption, true in our case, is crucial to apply decoupling.
\end{note}
Similarly we estimate the term $U_1$
\begin{lemma}[Bounding Linear Contributions]\label{lemma:linear_contrib}
For $\xi_i\sim^{iid} \mathsf{Bern}(p)-p$ and even $d\geqslant 2$ 
$$
\mathbf{E}|\sum_{i}\xi_i|^d \leqslant \mathbf{E}|\sum_{i} \eta_i |^d,\quad \eta_1,\eta_2\ldots \sim^{iid} \eta-\eta',\quad \eta,\eta'\sim \mathsf{Bern}(p)
$$
\end{lemma}
\begin{note}
For the proof we need only symmetrization.
\end{note}
Finally we reformulate the obtained bounds in terms of binomials
\begin{corollary}[Binomial Bounds for Linear and Quadratic Contributions]\label{cor:contrib_binom_bounds}
For $\xi_i\sim^{iid}\mathsf{Bern}(p)-p$, even $d\geqslant 2$, and $S,S'\sim^{iid} \mathsf{Binomial}(n,p)$ we have
$$
\mathbf{E}|\sum_{i}\xi_i|^d \leqslant \mathbf{E}(S-S')^{d},\quad
\mathbf{E}|\sum_{i\not=j}\xi_i\xi_j|^d \leqslant 16(\mathbf{E}(S-S')^{d})^2
$$
\end{corollary}

\begin{proof}[Proof of \Cref{cor:contrib_binom_bounds}]
The first bound follows directly as $\sum_{i}\eta_i$ is distributed as $S-S'$. To prove the second inequality, it suffices to show  for even $d\geqslant 2$ that
$$
\mathbf{E}(\sum_{i\not=j}\eta_{i}\eta'_j)^{d} \leqslant \mathbf{E}(\sum_{i\not=j}\eta_{i}\eta'_j + \sum_{i}\eta_i\eta'_i)^{d}
$$ 
because then $\mathbf{E}(\sum_{i\not=j}\eta_{i}+\sum_{i}\eta_i\eta'_i)^{d} = \mathbf{E} (\sum_{i}\eta_i)^d (\sum_{i}\eta'_i)^d = \mathbf{E}(\sum_{i}\eta_i)^d\cdot \mathbf{E}(\sum_{i}\eta'_i)^d $.

To prove the claim denote $A = \sum_{i\not=j}\eta_{i}\eta'_j $, $B=\sum_{i}\eta_i\eta'_i$ then we have to prove
$\mathbf{E}(A+B)^d \geqslant \mathbf{E}A^d$. Write $\mathbf{E}(A+B)^d = \sum_{k=0}^{d} \mathbf{E}[A^k B^{d-k}]$
and observe that $\mathbf{E}[A^k B^{d-k}]\geqslant 0$ which follows by expanding $A^k$ and $B^{d-k}$ into sums of products of 
$\eta_i,\eta'_i$ and utilizing their symmetry (terms with odd number of repetitions for some $\eta_i$ or $\eta'_i$ will have zero expectations, so we are left with non-negative square terms). Thus $\mathbf{E}(A+B)^d \geqslant \mathbf{E} A^d + \mathbf{E}B^d \geqslant \mathbf{E}A^d$.
\end{proof}
\begin{note}
An alternative bound avoids combinatorial argument and uses the triangle inequality to establish
$(\mathbf{E}|\sum_{i\not=j}\eta_i\eta_j|^d)^{1/d}\leqslant (\mathbf{E}|\sum_{i,j}\eta_i\eta'_j|^d)^{1/d}+(\mathbf{E}|\sum_{i=j}\eta_i\eta'_j|^d)^{1/d}$. The second term then behaves like a centered binomial with parameters $n,p^2$.
\end{note}

\begin{proof}[Proof of \Cref{lemma:quadratic_contrib}]
Write the $d$-th moment as $\mathbf{E}f(\sum_{i\not=j} \xi_i\xi_j)$ with convex $f(u) = |u|^d$.
By \Cref{prop:decouple} it is upper bounded by $4\mathbf{E}f(\sum_{i\not=j} \xi_i\xi'_j)$
where $\xi_i$ and $\xi'_i$ are indentically distributed and independent. Now look at some chosen $\xi_i$
and the expectation $\mathbf{E}f(\sum_{i\not=j} \xi_i\xi'_j)$ conditioned on the fixed values of the remaining variables (that is $\xi_j$ for $j\not=i$ and $\xi'_j$ for all $j$),
by \Cref{prop:symmetrize} we get that replacing $\xi_i$ by $\eta_i-\eta'_i$ where $\eta_i,\eta'_i$ are independent copies of $\xi_i$
gives an upper bound. We repeat this for all $\xi_i$ and the same for $\xi'_i$.
Note that each time we replace with the distribution $(\eta-p)-(\eta'-p) = \eta-\eta'$ where $\eta,\eta'\sim^{iid}\mathsf{Bern}(p)$.
\end{proof}
\begin{proof}[Proof of \Cref{lemma:linear_contrib}]
We replace $\xi_i$ iteratively as in the proof of \Cref{lemma:quadratic_contrib}.
\end{proof}



\subsubsection{Auxiliary Function}

The bounds in \Cref{lemma:symm_binom_moment} depends on expressions of form $a^{\ell}\ell^{b-\ell}$ that we analyze closer below
\begin{proposition}[Auxiliary Function]\label{prop:aux_func}
The function $g(\ell)\triangleq a^{\ell}\ell^{b-\ell}$, for any paramters $a,b>0$, is maximized at $\ell=\ell^{*}\triangleq b/W(b\mathrm{e}/a)$;
it increases for $0<\ell<\ell^{*}$ and decreases for $\ell^{*}<\ell<+\infty$.
where $W(\cdot)$ is the main branch of Lambert-W function.
\end{proposition}
\begin{proof}[Proof of \Cref{prop:aux_func}]
The derivative equals
$$
\frac{\partial g}{\partial \ell} = (a/\ell)^\ell \ell^{-1 + b} (b - \ell + \ell \log(a/\ell))
$$
and only the last factor can be possibly zero, therefore
$$
\frac{\partial g}{\partial \ell} = 0\Leftrightarrow   u\log u = b \mathrm{e}/a,\quad u\triangleq \mathrm{e}\ell/a 
$$
so the zero is at $u = \mathrm{e}^{W(b\mathrm{e}/a)} = b\mathrm{e} /a W(b/\mathrm{e}/a)$ and the formula for $\ell^{*}$ follows.
We also conclude that $g$ is monotone  in both invervals $0<\ell < \ell^{*}$ and $\ell^{*} < \ell < +\infty$.
By the equations above we see that $\frac{\partial g}{\partial \ell} > 0$ when $\ell\to 0$ and $\frac{\partial g}{\partial \ell} < 0$ when $\ell\to+\infty$, therefore we conclude that $g(\ell)$ increases for $0<\ell<\ell^{*}$ and decreases when $\ell^{*}<\ell<+\infty$.
\end{proof}

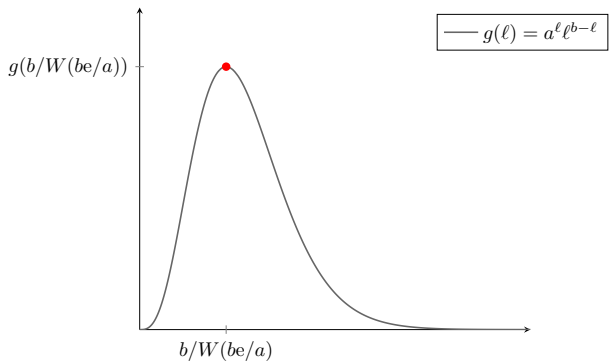
\begin{figure}[h!]
\begin{tikzpicture}[scale=0.75]
  \begin{axis}[
      samples=1001,
      enlarge y limits=true,
      axis lines=middle,
      xtick = \empty,
      extra x ticks = {1.1067},
      extra x tick labels = {$b/W(b\mathrm{e}/a)$},
      ytick = \empty,
      extra y ticks = {0.20409},
      extra y tick labels = {$g(b/W(b\mathrm{e}/a))$},
      ymax = 0.25,
     legend style={anchor=north}
    ]
\addplot [gray!80!black, thick, domain=0:5]  {(0.2/x)^x*x^3};
\addlegendentry { $g(\ell) =  a^{\ell}\ell^{b-\ell}$ };
\node [circle,fill=red,inner sep=1.5pt] (P) at (1.1067, 0.20409)  {};
  \end{axis}
\end{tikzpicture}
\caption{Auxiliary function $g(\ell)$ studied in \Cref{prop:aux_func}, here with parameters: $a=1/5,b=3$.}
\end{figure}

\begin{proposition}[Supremum of Auxiliary Function]\label{prop:aux_func_sup}
Let $g$ be as in \Cref{prop:aux_func}, then 
$$\sup\{g(\ell):1\leqslant \ell \leqslant b \} \leqslant  a\max(a,b)^{b-1}$$
\end{proposition}
\begin{proof}[Proof of \Cref{prop:aux_func_sup}]
Suppose that $a>b$, then by \Cref{prop:aux_func} we find that $g(\ell)$ is maximized at $\ell^{*} = b/W(b\mathrm{e}/a) > b $ (we use $W(u)<1$ iff $u<\mathrm{e}$);
then $g(\ell) \leqslant g(b)=a^b$ for $\ell \in [1,b]$.
If $a\leqslant b$ then $g(\ell) = a\cdot a^{\ell-1}\ell^{b-\ell} \leqslant a b^{\ell-1}\ell^{b-\ell} \leqslant a b^{b-1}$ when $\ell\in [1,b]$.
\end{proof}

\begin{proof}[Proof of \Cref{prop:aux_func_sup}]
Since $W(u) < 1$ when $0<u < \mathrm{e}$, we have $W(b\mathrm{e}/a) < 1$ when $a>b$ so 
$\ell^{*} > b$ and the supremum is at $\ell = b$.
Suppose now $a/b\geqslant 1$, then $W(b\mathrm{e}/a)\geqslant 1$ and 
$$
(a/\ell^{*})^{1/W}= W(b\mathrm{e}/a)^{1/W(b\mathrm{e}/a)}\cdot (a/b)^{W(b\mathrm{e}/a)} = W(u)^{1/W(u)} \cdot (u/\mathrm{e})^{1/W(u)},\quad u = b\mathrm{e}/a.
$$
Since $u\geqslant \mathrm{e}$ and $W(u)\geqslant 1$ we have $W(u)^{1/W(u)} = \Theta(1)$ and
$(u/\mathrm{e})^{1/W(u)}= \Theta(1)$ because for $u\gg 1$ we have $W(u) = \Theta(\log u)$. Therefore $(a/\ell^{*})^{1/W(u)}= \Theta(1)$ and
$$
g(\ell^{*}) = \Theta(1)^b (\ell^{*})^b = \Theta( b/W(u) )^b  \leqslant O(b)^b
$$
The result now follows, as the maximum can be either at $\ell=1$ or $\ell=\ell^{*}$ or $\ell=b$.
\end{proof}

\subsubsection{Moments of Bin Contributions}

Having estimated $U_1$ and $U_2$ in \Cref{sec:bin_contrib} we are in position to give formulas that control the moments of $\tilde{Q}_x$.
The exact bound is stated below
\begin{corollary}[Moments of Bin Contributions]\label{cor:bin_contrib_moments}
For every $x$ and $p=\Pr[X=x]$ we have
$$\mathbf{E} |\tilde{Q}_x-\mathbf{E}\tilde{Q}_x|^d  \leqslant O(d)^{d}(np)^2\max(d,np)^{d-2} + O(d)^{d/2}(np)^{d+1} \max(d,np)^{d/2-1}$$
\end{corollary}

\begin{remark}[Sub-Gamma Behavior]
When $np =\Omega(d)$ from \Cref{cor:bin_contrib_moments} we get that the behavior of $\tilde{Q}_x$ is \emph{sub-gamma} with variance factor $v^2= O((np)^2 + (np)^3)$ and scale $b = O(np)$.
\end{remark}

\begin{proof}[Proof of \Cref{cor:bin_contrib_moments}]
By \Cref{sec:bin_contrib} and the inequality $|a+b|^d \leqslant 2^{d-1}|a|^d + 2^{d-1}|b|^d$  
$$
\mathbf{E} |\tilde{Q}_x - \mathbf{E}\tilde{Q}_x|^d \leqslant 2^{d-1}\mathbf{E}U_2^d + 2^{d-1}(np)^d\mathbf{E}U_1^d
$$
for $p=\Pr[X=x]$ and $U_1$ and $U_2$ as defined there.
By \Cref{cor:contrib_binom_bounds} and \Cref{lemma:symm_binom_moment}
$$
\mathbf{E}|\tilde{Q}_x - \mathbf{E}\tilde{Q}_x|^{d} \leqslant 
\left( O(d)^{d/2}\sum_{\ell=1}^{d/2} \binom{n}{\ell} \ell^{d/2} \sigma^{2\ell}\right)^2 + O(d)^{d/2}(np)^{\ell}\sum_{\ell=1}^{d/2} \binom{n}{\ell} \ell^{d/2} \sigma^{2\ell}
$$
Overestimating $\left(\sum_{\ell=1}^{d/2} \binom{n}{\ell} \ell^{d/2} \sigma^{2\ell}\right)^2 \leqslant d/2\cdot \sum_{\ell=1}^{d/2} (\binom{n}{\ell} \ell^{d/2} \sigma^{2\ell})^2$
(Jensen's inequality) and $\binom{n}{\ell} \leqslant (n\mathrm{e}/\ell)^{\ell}$ (the well-known binomial bound)
and $\sigma^2 \leqslant 2p$ we obtain
$$
\mathbf{E}|\tilde{Q}_x - \mathbf{E}\tilde{Q}_x|^{d} \leqslant 
O(d)^{d/2}\sum_{\ell=1}^{d/2} g(\ell)^2 + O(d)^{d/2}(np)^{d}\sum_{\ell=1}^{d/2} g(\ell),\quad g(\ell)\triangleq (np)^{\ell}\ell^{d/2-\ell}.
$$
The result follows now from \Cref{prop:aux_func_sup}.
\end{proof}

\subsection{Assembling Bin Concentrations}

Armed with \Cref{cor:bin_contrib_moments} we are in position to estimate concentration of the sum of bin contribitions, and therefore the tails of the estimator. To this end we distinguish between
tails heavier and lighter than gamma. For fixed even $d$ we define light and heavy bins as
$$
\mathcal{X}^{-} \triangleq \{x:  n \Pr[X=x]  \geqslant d\},\quad \mathcal{X}^{+} \triangleq \{x:  n \Pr[X=x]  < d\}
$$
Now \Cref{thm:main} follows directly from the following facts
\begin{lemma}[Concentration of Contributions with Light Tails]\label{lemma:subgamma_conc}
We have
$$
\Pr[|\sum_{x\in\mathcal{X}^{-}} (\tilde{Q}_x-\mathbf{E}\tilde{Q}_x)|>\epsilon] \leqslant 2\exp(-\Omega(\epsilon^2/(v^2+b\epsilon)))
=2\exp(-\Omega(\min(\epsilon^2/v^2,\epsilon/b))).
$$
where $v^2 \triangleq  n^2\sum_x\Pr[X=x]^2 +  n^3\sum_x\Pr[X=x]^3,\quad b \triangleq n \max_{x}\Pr[X=x]$.
\end{lemma}

\begin{lemma}[Concentration of Contributions with Heavy Tails]\label{lemma:heavy_tails_conc}
We have
$$
\Pr[|\sum_{x\in\mathcal{X}^{+}} (\tilde{Q}_x-\mathbf{E}\tilde{Q}_x)|>\epsilon] \leqslant O(1)\exp(-\Omega(\sqrt{\epsilon})), \quad \epsilon^2 = \Omega(n^2\sum_{x}\Pr[X=x]^2).
$$
\end{lemma}

\begin{remark}[Concluding \Cref{thm:main}]
To conclude \Cref{thm:main} use
$\Pr[|X_1+X_2|\geqslant \epsilon] \leqslant \Pr[|X_1|\geqslant \epsilon/2] + \Pr[|X_2|\geqslant \epsilon/2]$
to combine tails contributed by $\mathcal{X}^{+}$ and $\mathcal{X}^{-}$. This gives
$$
Pr[|\sum_{x\in\mathcal{X}} (\tilde{Q}_x-\mathbf{E}\tilde{Q}_x)|>\epsilon] \leqslant O(1)\exp(-\Omega(\min(\epsilon^2,\epsilon/b,\sqrt{\epsilon})), \quad \epsilon = \Omega(v)
$$
The condition $\epsilon = \Omega(v)$ can be ignored because of the term $\epsilon^2/v^2$, this regime gives the trivial bound of $O(1)$.
The result follows by $Q = \frac{1}{n(n-1)}\sum_{x}\tilde{Q}_x$
 and changing $\epsilon:=n^2\epsilon$.
\end{remark}

\begin{proof}[Proof of \Cref{lemma:subgamma_conc}]
The result follows directly from known aggregation of sub-gamma distributions $\tilde{Q}_x$ which have paramters
$v^2=\sum_x v^2_x$ and $b= \max_x b_x$ as given in~\Cref{cor:bin_contrib_moments}. Aggregation is based on calculating moment generating functions, thus it does apply for negatively dependent random variables.
\end{proof}

\begin{proof}[Proof of \Cref{lemma:heavy_tails_conc}]
Observe that, by~\Cref{cor:bin_contrib_moments} applied to $d$ replaced with $k$, we have
$$\mathbf{E}(\tilde{Q}_x-\mathbf{E}\tilde{Q}_x)^k\leqslant O(k)^{2k-2} n^2\Pr[X=x]^2,\quad k=2,4,\ldots,d,\quad x\in\mathcal{X}^{+}$$
Thus
$$
\mathbf{E} (1+(\tilde{Q}_x-\mathbf{E}\tilde{Q}_x)/T)^{d} \leqslant 1+n^2\Pr[X=x]^2/T^2 \sum_{k=2}^{d} \binom{d}{k} O(k)^{2k-2} / T^{k-2}.
$$
We now apply \Cref{prop:latala_simple}, since $\phi(u) = \sum_{k=2\ldots d,k\text{ even}}\binom{d}{k}u^k$ we need to find $T$ such that
$$
n^2 \sum_{x\in\mathcal{X}^{+}}\Pr[X=x]^2/T^2 \sum_{k=2\ldots d, k \text{ even}}  \binom{d}{k} O(k)^{2k-2} / T^{k-2} \leqslant d
$$
This holds for $T = \Theta( v / \sqrt{d} + d^2)$, where $v^2=n^2\sum_{x\in\mathcal{X}^{+}}\Pr[X=x]^2$
 with appropriate constants. If this is the case then
$$\mathbf{E}|\sum_{x\in\mathcal{X}^{+}}(\tilde{Q}_x-\mathbf{E}\tilde{Q}_x)|^d\leqslant O(d^2 + v / \sqrt{d})^d$$
By Markov's inequality we obtain
$$
\Pr[ |\sum_{x\in\mathcal{X}^{+}}(\tilde{Q}_x-\mathbf{E}\tilde{Q}_x)|>\epsilon]\leqslant O((d^2/\epsilon + v/\sqrt{d}\epsilon)^d
$$
We set $d$ so that $d^2 = O(\epsilon)$ and $v/\sqrt{d} = O(\epsilon)$ with appropriate constants, which gives the tail of $2^{-d}$.
Note that $d$ has to be even and at least $2$, thus we the tail is $2^{-\Omega(\epsilon^{1/2})}$
if $\epsilon = \Omega(1)$ and $\epsilon = \Omega(v)$; the first condition can be ignored because the bound is then $O(1)$.
\end{proof}

\section{Conclusion}\label{sec:conclude}

We have derived strong concentration guarantees for the collision estimator, which subsumes variance bounds from previous works.
Such concentration bounds can be used for example to eliminate the need for boosting of weak estimators (majority/median tricks).

\bibliography{citations}

\begin{thebibliography}{10}

\bibitem{acharya2014complexity}
Jayadev Acharya, Alon Orlitsky, Ananda~Theertha Suresh, and Himanshu Tyagi.
\newblock The complexity of estimating r{\'e}nyi entropy.
\newblock In {\em Proceedings of the twenty-sixth annual ACM-SIAM symposium on
  Discrete algorithms}, pages 1855--1869. SIAM, 2014.

\bibitem{acharya2016estimating}
Jayadev Acharya, Alon Orlitsky, Ananda~Theertha Suresh, and Himanshu Tyagi.
\newblock Estimating r{\'e}nyi entropy of discrete distributions.
\newblock {\em IEEE Transactions on Information Theory}, 63(1):38--56, 2016.

\bibitem{subg_lecture_notes}
Gomel Amish and S.~Dey Partha.
\newblock Lecture notes in concentration inequalities, 2019.
\newblock \url{https://faculty.math.illinois.edu/~psdey/math595fa19/lec03.pdf}.

\bibitem{batu2001testing}
Tugkan Batu, Eldar Fischer, Lance Fortnow, Ravi Kumar, Ronitt Rubinfeld, and
  Patrick White.
\newblock Testing random variables for independence and identity.
\newblock In {\em Proceedings 42nd IEEE Symposium on Foundations of Computer
  Science}, pages 442--451. IEEE, 2001.

\bibitem{bellec2019concentration}
Pierre~C Bellec.
\newblock Concentration of quadratic forms under a bernstein moment assumption.
\newblock {\em arXiv preprint arXiv:1901.08736}, 2019.

\bibitem{boucheron2013concentration}
St{\'e}phane Boucheron, G{\'a}bor Lugosi, and Pascal Massart.
\newblock {\em Concentration inequalities: A nonasymptotic theory of
  independence}.
\newblock Oxford university press, 2013.

\bibitem{chernoff1952measure}
Herman Chernoff et~al.
\newblock A measure of asymptotic efficiency for tests of a hypothesis based on
  the sum of observations.
\newblock {\em The Annals of Mathematical Statistics}, 23(4):493--507, 1952.

\bibitem{cramer1938nouveau}
Harald Cram{\'e}r.
\newblock Sur un nouveau th{\'e}oreme-limite de la th{\'e}orie des
  probabilit{\'e}s.
\newblock {\em Actual. Sci. Ind.}, 736:5--23, 1938.

\bibitem{diakonikolas2016collision}
Ilias Diakonikolas, Themis Gouleakis, John Peebles, and Eric Price.
\newblock Collision-based testers are optimal for uniformity and closeness.
\newblock {\em arXiv preprint arXiv:1611.03579}, 2016.

\bibitem{dodis2013overcoming}
Yevgeniy Dodis and Yu~Yu.
\newblock Overcoming weak expectations.
\newblock In {\em Theory of Cryptography Conference}, pages 1--22. Springer,
  2013.

\bibitem{dubhashi1996balls}
Devdatt~P Dubhashi and Desh Ranjan.
\newblock Balls and bins: A study in negative dependence.
\newblock {\em BRICS Report Series}, 3(25), 1996.
\newblock \url{https://www.brics.dk/RS/96/25/BRICS-RS-96-25.pdf}.

\bibitem{goldreich2017introduction}
Oded Goldreich.
\newblock {\em Introduction to property testing}.
\newblock Cambridge University Press, 2017.

\bibitem{goldreich2011testing}
Oded Goldreich and Dana Ron.
\newblock On testing expansion in bounded-degree graphs.
\newblock In {\em Studies in Complexity and Cryptography. Miscellanea on the
  Interplay between Randomness and Computation}, pages 68--75. Springer, 2011.

\bibitem{joag1983negative}
Kumar Joag-Dev and Frank Proschan.
\newblock Negative association of random variables with applications.
\newblock {\em The Annals of Statistics}, pages 286--295, 1983.

\bibitem{latala1997estimation}
Rafa{\l} Lata{\l}a et~al.
\newblock Estimation of moments of sums of independent real random variables.
\newblock {\em The Annals of Probability}, 25(3):1502--1513, 1997.
\newblock \url{https://projecteuclid.org/download/pdf_1/euclid.aop/1024404522}.

\bibitem{obremski_et_al:LIPIcs:2017:7569}
Maciej Obremski and Maciej Skorski.
\newblock {Renyi Entropy Estimation Revisited}.
\newblock In Klaus Jansen, Jos{\'e} D.~P. Rolim, David Williamson, and
  Santosh~S. Vempala, editors, {\em Approximation, Randomization, and
  Combinatorial Optimization. Algorithms and Techniques (APPROX/RANDOM 2017)},
  volume~81 of {\em Leibniz International Proceedings in Informatics (LIPIcs)},
  pages 20:1--20:15, Dagstuhl, Germany, 2017. Schloss Dagstuhl--Leibniz-Zentrum
  fuer Informatik.
\newblock URL: \url{http://drops.dagstuhl.de/opus/volltexte/2017/7569}, \href
  {http://dx.doi.org/10.4230/LIPIcs.APPROX-RANDOM.2017.20}
  {\path{doi:10.4230/LIPIcs.APPROX-RANDOM.2017.20}}.

\bibitem{ortiz2003concentration}
Luis~E Ortiz and David~A McAllester.
\newblock Concentration inequalities for the missing mass and for histogram
  rule error.
\newblock In {\em Advances in Neural Information Processing Systems}, pages
  367--374, 2003.

\bibitem{paninski2008coincidence}
Liam Paninski.
\newblock A coincidence-based test for uniformity given very sparsely sampled
  discrete data.
\newblock {\em IEEE Transactions on Information Theory}, 54(10):4750--4755,
  2008.

\bibitem{rosenthal1970subspaces}
Haskell~P Rosenthal.
\newblock On the subspaces ofl p (p> 2) spanned by sequences of independent
  random variables.
\newblock {\em Israel Journal of Mathematics}, 8(3):273--303, 1970.

\bibitem{rudelson2013hanson}
Mark Rudelson, Roman Vershynin, et~al.
\newblock Hanson-wright inequality and sub-gaussian concentration.
\newblock {\em Electronic Communications in Probability}, 18, 2013.

\bibitem{vershynin2018high}
Roman Vershynin.
\newblock {\em High-dimensional probability: An introduction with applications
  in data science}, volume~47.
\newblock Cambridge university press, 2018.

\end{thebibliography}

\end{document}